\documentclass[11pt]{article}
\usepackage[margin=1in]{geometry}
\usepackage{algorithm}
\usepackage{amsmath,mathtools}
\usepackage{color}
\usepackage{amsfonts}
\usepackage{amssymb}
\usepackage{xspace}
\usepackage{comment}
\usepackage{hyperref}
\usepackage{float}
\usepackage{amsmath}
\usepackage{algpseudocode}

\usepackage{subfigure}
\usepackage{times}
\usepackage{sidecap}
\usepackage{enumerate}
\usepackage{tcolorbox}

\usepackage{blindtext}
\usepackage{enumitem}
\usepackage{xcolor}
\def \polylog{\operatorname{polylog}}
\usepackage[symbol]{footmisc}
\usepackage{varioref}
\usepackage{array}
\newcolumntype{P}[1]{>{\centering\arraybackslash}p{#1}}
\newcolumntype{M}[1]{>{\centering\arraybackslash}m{#1}}

\usepackage{amsmath,amssymb, amsthm, enumerate,boxedminipage}

\theoremstyle{definition}
\newtheorem*{def*}{Definition}
\newtheorem{definition}{Definition}
\newtheorem{theorem}{Theorem}[section]
\theoremstyle{definition}

\newtheorem{lemma}{Lemma}

\usepackage{mathtools}

\theoremstyle{definition}
\newtheorem*{prf*}{Proof}
  
\theoremstyle{definition}
\newtheorem*{prfthm*}{Proof of Theorem}

\newtheorem{remark}{Remark}

\newcommand{\I}{\mathcal{I}}

\newcommand{\R}{\mathcal{R}}

\newcommand{\eps}{\epsilon}
\newcommand{\anis}[1]{{\color{green}\underline{\textsf{A:}}} {\color{blue} \emph{#1}}}
\newcommand{\manish}[1]{{\color{blue}\underline{\textsf{M:}}} {\color{red} \emph{#1}}}

\newboolean{short}
\setboolean{short}{true}
\newcommand{\shortOnly}[1]{\ifthenelse{\boolean{short}}{#1}{}}
\newcommand{\onlyShort}[1]{\ifthenelse{\boolean{short}}{#1}{}}
\newcommand{\longOnly}[1]{\ifthenelse{\boolean{short}}{}{#1}}
\newcommand{\onlyLong}[1]{\ifthenelse{\boolean{short}}{}{#1}}
\newcommand{\shortLong}[2]{\ifthenelse{\boolean{short}}{#2}{#1}}
\newcommand{\longShort}[2]{\ifthenelse{\boolean{short}}{#2}{#1}} 

\onlyShort{
\usepackage{enumitem}
\setitemize{noitemsep,topsep=0pt,parsep=0pt,partopsep=0pt}
\setenumerate{noitemsep,topsep=0pt,parsep=0pt,partopsep=0pt}
}
\begin{document}
\title{Sublinear Message Bounds of Authenticated Implicit Byzantine Agreement\thanks{The work of A. R. Molla was supported in part by ISI DCSW Project, file number E5413.}}

\author{Manish Kumar \thanks{Indian Statistical Institute, Kolkata 700108, India.\hbox{E-mail}:~{\tt manishsky27@gmail.com}.}
\and Anisur Rahaman Molla \thanks{Indian Statistical Institute, Kolkata 700108, India.  \hbox{E-mail}:~{\tt anisurpm@gmail.com}.}}
\date{}

\maketitle \thispagestyle{empty}

\maketitle
\begin{abstract}
This paper studies the message complexity of authenticated Byzantine agreement (BA) in synchronous, fully-connected distributed networks under an honest majority. We focus on the so-called {\em implicit} Byzantine agreement problem where each node starts with an input value and at the end a non-empty subset of the honest nodes should agree on a common input value by satisfying the BA properties (i.e., there can be undecided nodes)\footnote{Implicit BA is a generalization of the classical BA problem. Throughout, we write Byzantine agreement or BA to mean implicit Byzantine agreement.}. We show that a sublinear (in $n$, number of nodes) message complexity BA protocol under honest majority is possible in the standard PKI model when the nodes have access to an unbiased global coin and hash function. In particular, we present a randomized Byzantine agreement algorithm which, with high probability achieves implicit agreement, uses $\tilde{O}(\sqrt{n})$  messages, and runs in $\tilde{O}(1)$ rounds while tolerating  $(1/2 - \epsilon)n$  Byzantine nodes for any fixed $\epsilon > 0$, the notation $\Tilde{O}$ hides a $O(\polylog{n})$ factor\footnote{We use the abbreviated phrase ``w.h.p." or ``with high probability," signifying a probability of at least $1-1/n^{\epsilon}$ for a constant $\epsilon$, where $n$ corresponds to the number of nodes present in the network.}. The algorithm requires  standard cryptographic setup PKI and hash function with a static Byzantine adversary. The algorithm works in the CONGEST model and each node does not need to know the identity of its neighbors, i.e., works in the $KT_0$ model. The message complexity (and also the time complexity) of our algorithm is optimal up to a $\polylog n$ factor, as we show a $\Omega(\sqrt{n})$ lower bound on the message complexity. \onlyLong{We also perform experimental evaluation and highlight the effectiveness and efficiency of our algorithm. The experimental results outperform the theoretical guarantees.
We further extend the result to Byzantine subset agreement, where a non-empty subset of nodes should agree on a common value. We analyze several relevant results which follow from the construction of the main result.}      

To the best of our knowledge, this is the first sublinear message complexity result of Byzantine agreement. A quadratic message lower bound is known  for any deterministic BA protocol (due to Dolev-Reischuk [JACM 1985]). The existing randomized BA protocols have at least quadratic message complexity in the honest majority setting. Our result shows the power of a global coin in achieving significant improvement over the existing results. It can be viewed as a step towards understanding the message complexity of randomized Byzantine agreement in distributed networks with PKI. 
\end{abstract}

\noindent {\bf Keywords:} Distributed Algorithm, Randomized Algorithm, Byzantine Agreement, Message Complexity, Global Coin, Cryptographic Assumptions.
%\vspace{-0.5cm}
\section{Introduction}\label{sec:introduction}
%\vspace{-0.2cm}
Byzantine agreement is a fundamental and long studied problem in distributed networks \cite{LSP82,AW,LynchBook}. In this problem, all the nodes are initiated with an input value. The Byzantine agreement problem is required to satisfy: (i) the honest nodes must decide on the same input value; and (ii) if all the honest nodes receive the same input value, then they must decide on that value\footnote{Throughout, we interchangeably use the term `non-Byzantine' and `honest', and similarly, `Byzantine' and `faulty'.}. This should be done in the presence of a constant fraction of Byzantine nodes that can arbitrarily deviate from the protocol executed by the honest nodes. Byzantine agreement provides a critical building block for creating attack-resistant distributed systems. Its importance can be seen from widespread and continued application in many domains such as wireless networks~\cite{kong2004anonymous,li2009privacy,weber2010internet,sicari2015security}, sensor networks \cite{SP04}, grid computing \cite{AK02}, peer-to-peer networks \cite{REGWZK03} and cloud computing \cite{Wright09}, cryptocurrencies~\cite{bitcoin,ethereum,AMN0S17,KJGKGF16,Micali16}, secure multi-party computation \cite{BGW88} etc.  However, despite huge research, we lack efficient practical solutions to Byzantine agreement for large networks. A main drawback for this is the {\em large message complexity} of currently known protocols, as mentioned by many systems papers \cite{AF03,ADK06,castro2002practical,MR97,YHET05}. The best known Byzantine protocols have (at least) quadratic message complexity \cite{vinod1,vinod2,feldman,kowalski}, even in the authenticated settings \cite{DS83,MR20}. In a distributed network, nodes communicate with their neighbors by passing messages. Therefore, communication cost plays an important role to analyze the performance of the algorithms, as also mentioned in many papers~\cite{abraham2019communication,halpern,gray,MR20}.

King and Saia \cite{KS11} presented the first Byzantine agreement algorithm that {\em breaks} the quadratic message barrier in synchronous, complete networks. The message complexity of their algorithm is $\tilde{O}(n^{1.5})$. Later, Braud-Santoni et al.~\cite{SGH13} improved this to $\tilde{O}(n)$ message complexity. Both works require the nodes to know the IDs of the other nodes a priori. This model  is known as
$KT_1$ model (known till 1)~\cite{Pelege2000}. Another challenging model is $KT_0$, where nodes do not know their neighbors a priori \cite{Pelege2000}. Note that in $KT_0$ model, nodes can know their neighbors easily by communicating to all the neighbors, perhaps in a single round, but that will incur $\Omega(n^2)$ messages.  The $KT_0$ model is more appropriate to the modern distributed networks which are permissionless, i.e., nodes can enter and leave the network at will.

In this paper, our main focus is to study the message complexity of the Byzantine agreement problem in the $KT_0$ model under the assumption of cryptographic setup and a global coin (as defined in \cite{Rabin83}). In fact, we study the implicit version of the Byzantine agreement, where not all the honest nodes need to be decided; only a non-empty subset of the honest node must decide on an input value. Our main result is a randomized algorithm to solve implicit Byzantine agreement using sublinear messages (only $\tilde{O}\sqrt{n}$) while tolerating $f \le (1/2-\eps)n$ Byzantine nodes, where $n$ is the number of nodes in the network, $f$ is the number of Byzantine nodes and $\eps>0$ is a fixed constant. The implicit algorithm can be easily extended to the explicit Byzantine agreement (where all the honest nodes must decide) using $O(n\log n)$ messages only. The algorithm is simple and easily implementable, which is highly desired for practical purposes. While the assumptions on the Public Key Infrastructure (PKI) set up with keyed hash function and the global coin together make the model a little weaker, they are realistic and implementable \footnote{Throughout, we interchangeably use the term `hash function' and `keyed hash function'.}. Similar assumptions were made earlier in the literature, e.g., in Algorand, Gilad et al. \cite{GHMVZ17} uses “seed” and PKI setup, where “seed” is essentially the shared random bits. In their approach, they formed a set of candidate nodes and reached an agreement with the help of the sortition algorithm (see Section 5 of \cite{GHMVZ17}) using proof-of-stake (PoS). They further used verifiable random functions (VRFs) \cite{MRV99} for the verification of the candidate nodes which return hash and proof. In our work, we do not require the assumption of VRFs. The message complexity of Algorand is $\tilde{O}(n)$, albeit for the explicit agreement. 
%\anis{compare the results...although the model is different, algorand may take sublinear messages for the implicit agreement by combining our sampling idea after committee election.} \manish{In Algorand message complexity is $\tilde{O}(n)$ while our finding have the message complexity $\tilde{O}(\sqrt{n})$.} \manish{There are some differences in our model and Algorand like i) based on weighted fraction more than $2/3$ nodes are honest while in our model honest nodes are in majority i.e., $(1/2+\eps)$. ii) Algorand's adversary is also $KT_0$ while we are using rushing adversary. iii) Algorand is using gossip protocol. iv) Input value can be varified whether exist or not unlike our model where we have to rely on majority for agreement as well as existance of a value. v) Algorand is semi synchronous while our model is synchronous.} 

Without PKI setup, hash function and global coin assumptions, we do not know if a sublinear (or even a linear) message complexity Byzantine agreement algorithm is possible or not. So far, the best results have quadratic message bound in the $KT_0$ model and sub-quadratic in the $KT_1$ model. %We believe that the topic would be interesting to the audience (in the security aspects of distributed computing) due to its direct application in cryptocurrencies and permission-less distributed systems. 

%To the best of our knowledge, 
Our result introduces the first {\em sublinear message complexity} Byzantine agreement algorithm and at the same time tolerates optimal resilience, i.e., $f \le (1/2-\eps)n$. \onlyLong{We further extend the implicit agreement to a natural generalized problem, called {\em subset agreement} problem. The subset agreement problem can be useful in real applications. For example, in a large scale distributed network, it may require that a non-empty subset of the nodes (unknown to each other) want to agree on a common value. Since, typically, the size of the subset is much smaller than the network size, the cost of the agreement would be less than the explicit agreement. Thus, subset agreement may work as a subroutine in many applications.} We also argue a lower bound on the message complexity of the problem. The lower bound shows that the message complexity of our algorithm is optimal up to a $\polylog n$ factor. \onlyLong{Finally, we implement our algorithm to evaluate its actual performance.} Our results can be viewed as a step towards understanding the message complexity of randomized BA in distributed networks under the assumptions of PKI, hash function and global coin.

\medskip
%\vspace{-0.1cm}
\noindent \textbf{Paper Organization:} The rest of the paper is organized as follows. In the rest of this section, we state our result and introduce the model and definition. Section \ref{sec:relatedwork} is a related work, which introduces the seminal works done in the same direction. Section \ref{sec:agreement} presents the main implicit Byzantine agreement algorithm\onlyLong{and other relevant results}. \onlyLong{Section~\ref{sec:subset agreement} presents the Byzantine subset agreement.} In Section ~\ref{sec: lower_bound}, we present the lower bound on the message complexity to support the optimality of our algorithm. \onlyLong{Section~\ref{sec: experiment} shows the experimental evaluation of the implicit Byzantine agreement algorithm.} Finally, we conclude with some open problems in Section~\ref{sec:conclusion}.

%\vspace{-0.4cm}
\subsection{Our Results}
%\vspace{-0.1cm}
We show the following main results.  
%\vspace{-0.1cm}
\begin{theorem}[Implicit Agreement]\label{main_theorem}
Consider a synchronous, fully-connected, anonymous network of $n$ nodes and CONGEST communication model. Assuming a public-key infrastructure with the keyed hash function, there exists a randomized algorithm which, with the help of global coin, solves implicit Byzantine agreement with high probability in $O(\log^2 n)$ rounds and uses $O(n^{0.5} \log^{3.5} n)$ messages while tolerating  $f \le (1/2-\epsilon)n$ Byzantine nodes under non-adaptive adversary, where $\eps$ is any fixed positive constant.  
\end{theorem}
%\vspace{-0.3cm}
\begin{theorem}[Explicit Agreement]
Consider a synchronous, fully-connected network of $n$ nodes and CONGEST communication model. Assuming a public-key infrastructure with the keyed hash function, there exists a randomized algorithm which, with the help of global coin, solves Byzantine agreement with high probability in $O(\log^2 n)$ rounds and uses $O(n \log n)$ messages while tolerating  $f \le (1/2-\epsilon)n$ Byzantine nodes under non-adaptive adversary, where $\eps$ is any fixed positive constant.   
\end{theorem}

%put here 
\onlyLong{
\begin{theorem}[Byzantine Subset Agreement]
Consider a complete $n$-node network $G (V, E)$ and a subset $S\subseteq V$ of size $k$. There is a randomized algorithm which, with the help of a global coin and PKI set up with keyed hash function, solves the Byzantine subset agreement over $S$ with high probability and finishes in $\Tilde{O}(k)$ rounds and uses $\min\{ \Tilde{O}(k\sqrt{n}), \Tilde{O}(n)\}$ messages. %\manish{In worst case our algorithm takes $O(\log^2n)$ rounds. I think round should be $ \Tilde{O}(1)$} \anis{it's subset agreement...round would $\Tilde{O}(k)$ ---right?} \manish{since the subset size is known. Therefore, a small comittee can be choosen randomly with the help of global coin  (as we did in main  algo) and subset can reach at implicit agreement. After that remaining nodes of the subset reach at agreement via referee nodes (depends on the situation).}
\end{theorem}
}
%\vspace{-0.3cm}
\begin{theorem}[Lower Bound]
Consider any algorithm $A$ that has access to an unbiased global coin and sends at most $f(n)$ messages (of arbitrary size) with high probability on a complete network of $n$ nodes. If $A$ solves the authenticated Byzantine agreement under honest majority with constant probability, then $f(n) \in \Omega(\sqrt{n})$. 
\end{theorem}

\onlyLong{Finally, we implement our implicit agreement algorithm and show its effectiveness and efficiency for different sizes of Byzantine nodes. When the Byzantine nodes behave randomly, the experimental results perform better than the theoretical guarantees.}

\subsection{Model and Definitions}\label{sec:model}
%\vspace{-0.05cm}
The network is synchronous and fully-connected graph of $n$ nodes. Initially, nodes do not know their neighbors, also known as $KT_0$ model \cite{Pelege2000}. Nodes have access to an unbiased global coin through which they can generate shared random bits. The network is $f \leq (1/2-\epsilon)n$ resilient, i.e., at most $(1/2-\epsilon)n$ nodes (among $n$ nodes) could be faulty, for any constant $\eps>0$. We consider Byzantine fault \cite{LSP82}. A Byzantine faulty node can behave maliciously such that it sends any arbitrary message or no message in any round to mislead the protocol, e.g., it may send different input values to different nodes, or it may not send any message to some of the nodes in a particular round. %A network is called $f$-resilient if at most $f$ nodes are faulty (we consider Byzantine fault \cite{LSP82}). 
We assume that a {\em static} adversary controls the Byzantine nodes, which selects the faulty nodes before the execution starts. However, the adversary can adaptively choose when and how a node behaves maliciously. Further, the adversary is rushing and has full information-- the adversary knows the states of all the nodes and can view all the messages in a round before sending out its own messages for that round. We assume that each node possesses multi-valued input of size $O(\log n)$ provided by the adversary.

We assume the existence of digital signatures, Public Key Infrastructure (PKI) and hash function. Each node possesses a public-secret key pair ($p_{k}, s_{k}$). Secret keys are generated randomly, and correspondingly public keys are generated with the help of a prime order group's generator. Therefore, the public keys are not skewed but random. Trusted authority also provides other cryptographic primitives for each node, and certifies each node’s public keys. 
%A threshold signature scheme \cite{CKPS01, LJY16} can be used. In the threshold signature scheme, a set of signatures $M_s$ (message $M$ is signed by a node $s$) from $t$  distinct nodes (threshold) can be combined into a threshold signature for $M$ with the same length as an individual signature. The currently known threshold signature schemes require a trusted dealer who generates all public and private keys for all nodes and a group of public keys to verify an aggregated full signature (henceforth it is called a trusted setup). 
Nodes use digital signatures for the authentication of any information. We abstract away the details of the cryptographic setup; assuming it is a standard framework. Public key ($p_k$) of all the nodes, hash function, shared random bits (generated through global coin) and $n$ are the common knowledge for all the nodes. %Global coin is the randomized shared global value which is provided after starting the execution.

We consider the {\em CONGEST} communication model \cite{Pelege2000}, where a node is allowed to send a message of size (typically) $O(\log n)$ or $O(\polylog (n))$ bits through an edge per round. The message complexity of an algorithm is the total number of messages sent by all the non-faulty nodes throughout the execution of the algorithm. 

%As security parameter $\kappa$ is in the order of input size, the value of $n$ might be very large and we are using CONGEST communication model. Therefore, we are considering the security parameter of the order of $O(\log n)$. The message complexity of an algorithm is the total number of messages sent by all the non-faulty nodes throughout the execution of the algorithm. %We assume that nodes know the value of $p_k$ and $n$. Although a node knows the value of $n$  simply by its number of ports connected to the other nodes since the network is complete.
\medskip
%\noindent \textbf{Byzantine Agreement vs. Byzantine Broadcast.}
%\vspace{-0.35cm}
\begin{definition}[Implicit Byzantine Agreement]\label{def:implicitBA} 
Suppose initially all the nodes have an input value (say, provided by an adversary).  An implicit Byzantine agreement holds when the following properties hold: (i) the final state of all the non-Byzantine nodes is either “decided'' or “undecided''; (ii) all the “decided'' non-Byzantine nodes must agree on the same value (consistency property); (iii) if all the non-Byzantine nodes have the same input value then they must decide on that value (validity property); (iv) all the non-Byzantine nodes eventually reach to the final state, where at least one non-Byzantine node must be in the “decided'' state (termination).
\end{definition}

\onlyLong{
\begin{definition}[Byzantine Subset Agreement]\label{def:subsetBA}
Suppose initially all the nodes have an input value (say, provided by an adversary). The Byzantine subset agreement is an agreement by a (specified) non-empty subset $S \subseteq V$ of nodes such that the fraction of Byzantine nodes in $S$ is preserved, i.e., $f_S \le (1/2 - \eps)|S|$, where $f_S$ is the number of Byzantine nodes in $S$.  We assume that each node knows whether it belongs to $S$ or not, but does not know the identities of the other nodes in the subset. The agreement on $S$ holds when the final state of all the non-faulty nodes in $S$ is “decided” and the deciding value of all of them follows the properties: (i) all the non-faulty nodes must decide on the same value (consistency property); (ii) if all the non-faulty nodes of network have the same initial value then they (the non-faulty nodes in $S$) must decide on that value (validity property); (iii)  all the non-faulty nodes (of $S$) eventually decide on a value (termination property).

\end{definition}
}
%\vspace{-0.45cm}
\begin{definition}[Keyed Hash Function] 
Let $\mathcal{K}_h, X$ be two non-empty finite sets and $H$ be an $b$-bit function such that $H : \mathcal{K}_h \times X \rightarrow \{0, 1\}^b$, where $H$ follows three properties: (i) Preimage Resistant - Given a hash value $h$, it should be difficult to find any message $m$ such that $h = H(k,m)$. (ii) Second Preimage Resistant - Given an input $m_1$, it should be difficult to find a different input $m_2$ such that $H(k,m_1) = H(k,m_2)$. (iii) Collision Resistant - It should be difficult to find two different messages $m_1$ and $m_2$ such that $H(k,m_1) = H(k,m_2)$.

\end{definition}

% \begin{definition} [Global Coin]

% \end{definition}
\onlyLong{
\subsection{Byzantine Agreement vs. Byzantine Broadcast}
Byzantine Agreement is typically studied in two forms while they are equivalent, i.e., one can be reduced to the other \cite{LSP82}. In the agreement version, also known as {\em Byzantine Consensus}, prior to the protocol starts, all the nodes receive an input value. Byzantine agreement is achieved if the following three properties hold.\\
Consistency: all of non-faulty nodes must decide on the same value. \\
%\manish{Or a non-faulty node does not decide on any value.}\\
Validity: if all the non-faulty nodes have the same initial value, then they must decide on that value.\\
Termination: all the non-faulty nodes eventually decide on a value.\\
%\manish{at least one}
}
%\vspace{-0.2cm}
In the {\em Byzantine Broadcast}, there is a designated sender (could be Byzantine or honest) who broadcasts the input values to the nodes. Termination and consistency are the same as in the case of Byzantine agreement. In the case of validity, all the non-faulty nodes output the same value if the sender is honest. Also, that value should be the input value of the sender.

%\vspace{-0.3cm}
\section{Related Work}\label{sec:relatedwork}
%\vspace{-0.2cm}
Byzantine agreement and Byzantine broadcast have been studied extensively in various models and settings in the last four decades, starting from its classic introduction by Lamport, Shostak and Pease  \cite{LSP82, PSL80}. They presented protocols and fault tolerance bounds for two settings (both synchronous). Without cryptographic assumptions (the unauthenticated setting), Byzantine broadcast and agreement can be solved if $f < n/3$. Assuming digital signatures (the authenticated setting), Byzantine broadcast can be solved if $f < n$ and Byzantine agreement can be solved if $f < n/2$. The initial protocols had exponential message complexities \cite{LSP82, PSL80}. Fully polynomial protocols were later shown for both the authenticated $(f < n/2)$ \cite{DS83} and the unauthenticated $(f < n/3)$ \cite{GM98} settings. Both protocols require $f + 1$ rounds of communication, which matches the lower bound on round complexity for deterministic protocols \cite{FL82}.

\begin{table*}[t]
\centering
%\begin{tabular}{|P{2.75cm}|P{2.00cm}|P{3.25cm}|P{2.25cm}|P{2cm}|P{3cm}|}
% \begin{tabular}{|P{2.75cm}|P{2.00cm}|P{3.25cm}|P{2.25cm}|P{1.75cm}|P{2.75cm}|}
\begin{tabular}
{|P{2.45cm}|P{2.00cm}|P{2.3cm}|P{2.25cm}|P{2.2cm}|P{2.75cm}|}
%{|P{2.35cm}|P{1.40cm}|P{2.0cm}|P{1.85cm}|P{1.8cm}|P{2.1cm}|}
%\begin{tabular}{|P{2.75cm}|P{2.90cm}|P{3.25cm}|P{2.25cm}|P{1.75cm}|P{2.75cm}|}
%\begin{tabular}{|P{2.35cm}|P{1.4cm}|P{2cm}|P{1.8cm}|P{1.8cm}|P{2.1cm}|}
\hline
\multicolumn{6}{|c|} {Comparison of the Results}\\
\hline
 Protocol & Agreement Type & Communication (in bits) & Adversary  & Cryptographic Assumptions & Resilience  \\
\hline

&&&&&\\
Santoni et al. \cite{SGH13}& Implicit    & $\Tilde{O}(n)$ & Non-adaptive  & No & $f< n/(3+\epsilon)$ \\
&&&&&\\
King-Saia \cite{KS11}& Explicit     & $\Tilde{O}(n^{1.5})$ & Adaptive  & No & $f\leq (1/3-\epsilon)n$ \\
&&&&&\\

%Berman et al. \cite{BGP92}& Explicit      & $O(n^2)$ & Adaptive  & No & $f < n/3$  \\
%&&&&&\\
%KT0         & Zero & Zero         & Yes & $f\leq (n-1)$ & Leader Election \\
%&&&&&\\
Dolev-Strong \cite{DS83}& Explicit      & $\Tilde{O}(n^3)$ & Adaptive  & Yes&  $f<n/2$ \\
&&&&&\\
%KT0      & $O(n \log^2n)$ & ---  & No & --- & Leader Election\\
%&&&&&\\
 Momose-Ren \cite{MR20}& Explicit      & $\Tilde{O}(n^2)$ & Adaptive  & Yes & $f< n/2$ \\
&&&&&\\
%KT1      & Zero & Zero  & - & - & Leader Election \\
%&&&&&\\
Abraham et al. \cite{ADDN019}& Explicit      & $O(n^2)^*$ & Adaptive  & Yes & $f< n/2$ \\
&&&&&\\
\textbf{This paper} & Implicit     & $ \Tilde{O}(n^{0.5})$& Non-adaptive  & Yes& $f\leq (1/2-\epsilon)n$ \\
&&&&&\\
\textbf{This paper} & Explicit     & $ \Tilde{O}(n)$& Non-adaptive  & Yes& $f\leq(1/2-\epsilon)n$ \\
&&&&&\\
\hline
\end{tabular}
%\caption{Our results assume the access of a global coin unlike others. $\epsilon$ is any positive constant.}
\caption{Comparison of various model with our result. $\epsilon$ is any positive constant. Our results assume a global coin and hash function while others are not. * indicates the bound holds in expectation.}
\label{tbl:related_work}
\end{table*}

 %In our model, we don't use {\em Verifiable Random Function} (VRF) unlike the seminal work of Gilad et al. \cite{GHMVZ17}. They used the global seed (random value) as the input of VRF which output hash and proof; it helps to select the small number of candidate nodes for the participating nodes. Similar results can be obtained by, using the techniques used in this paper, their work but the advantage of our scheme is $-$ there is no need of hash and proof. Also our protocol gives the advantage to select the fixed number of candidate nodes with probability one.

Our Byzantine agreement algorithm makes use of a few past results. First, we make use of the concept of a set of candidate nodes, which is a subset of nodes. The notion of a committee (set of candidate nodes) is  used in, e.g., \cite{AMP18, GK10, KM21, KPPRT15}. Finally, we adapt the Byzantine Agreement algorithm designed by Dolev et al. \cite{DS83}. %We tried to answer the question asked in the pliable work of Kumar-Molla \cite{KM21} where crash-fault message complexity was shown sub-linear. 
%Our work is highly inspired from the work of Kumar-Molla.

The previous results in the same direction are summarized in the Table \ref{tbl:related_work}. \onlyLong{Santoni et al. \cite{SGH13} achieved the $\Tilde{O}(n)$ communication complexity against a non-adaptive adversary while resilience is $f< n/(3+\epsilon)$ and rushing adversary with $KT_0$ model. The work of King-Saia \cite{KS11} and Abraham et al. \cite{ADDN019} used the shared random value in one or other way. King-Saia reached at almost everywhere to everywhere agreement in communication complexity (communication complexity of a protocol is the maximum number of bits sent by all the non-Byzantine nodes combined across all executions) $\Tilde{O}(n^{1.5})$ bits with adaptive security while resilience is $f<(1/3-\epsilon)n$ and in the  $KT_1$ model. Abraham et al. showed the quadratic  communication complexity in expectation with an adaptive adversary while used the cryptographic assumptions and tolerated $f<n/2$  Byzantine nodes. Dolev-Strong \cite{DS83} achieved cubic communication complexity with adaptive adversary and cryptographic assumptions by tolerating $f<n/2$ Byzantine nodes. Later, Momose-Ren \cite{MR20} improved the communication complexity to quadratic communication.}  %\footnote{\manish{In $KT_1$ model, we can also achieve the similar message complexity with the help of Section \ref{sec:agreement} $C.$ that too with rushing adversary. But there we use the standard assumption, the Byzantine nodes can not fake their IDs when sending the messages}}.  %King-Saia \cite{KS09} broke the the quadratic communication complexity but with non-adaptive adversary while tolerating $f \leq (1/3-\epsilon)n$ Byzantine nodes. 
In comparison with our work, we are using shared random value and non-adaptive adversary with rushing adversary. Shared random value help us to break the linear communication (message) complexity. In implicit agreement, we have sublinear communication complexity while in explicit agreement it is linear with cryptographic assumptions and tolerate $f \leq (1/2-\epsilon)n$ Byzantine nodes. 
%If there is no cryptographic assumptions then we have linear communication complexity while tolerating $f \leq (1/3-\epsilon)n$ Byzantine nodes. 
%The works of both Dolev-Strong \cite{DS83} and Momose-Ren \cite{MR20} achieve quadratic message complexity with an adaptive adversary and optimal resilience using cryptographic assumptions but Dolev-Strong uses cubic communication. King-Saia achieve the sub-quadraticc complexity.

%\anis{to myself: we are not sure if $\Omega(n^{.5})$ could be a lower bound, since there exists $n^{0.4}$ message bound agreement algorithm \cite{AMP18} without Byzantine nodes but with the access of a Global coin. The $n^{0.4}$ message bound algorithm won't work in the presence of Byzantine nodes with a linear resilience. Their algorithm is based on sampling input values (binary) and computing the fraction of 1s and decide on 0 or 1 based on the fraction value lies on the left or right of random value (generated through the Global coin) in the interval $[0, 1]$. The Byzantine nodes can easily take over this algorithm by giving wrong (biased) input values during sampling. Another point to note that the algorithm in \cite{AMP18} works only for binary agreement where the input values are 0, 1. Our algorithm works for multi-valued agreement.} %\manish{In case of multi-valued agreement message complexity might increase. Dolev-Strong algorithm is for binary input value.}

%\vspace{-0.25cm}
\section{Authenticated Implicit Byzantine Agreement}\label{sec:agreement}
\vspace{-0.15cm}
In this section, we present a randomized Byzantine agreement  algorithm in a complete $n$-node network that tolerates $f \leq (1/2-\epsilon)n$ Byzantine nodes under a public-key infrastructure, keyed hash function and access to a global coin. The algorithm incurs $\tilde{O}(\sqrt{n})$ messages, has latency $\tilde{O}(1)$, and has high success probability.  

In the algorithm, we run a subroutine BA protocol, which can tolerate $f \leq (1/2-\epsilon)n$ Byzantine nodes and may have a polynomial message (and time) complexity. In fact, we adapt the classical algorithm presented by Dolev-Strong \cite{DS83}.\footnote{One can use other suitable BA protocols, e.g., the protocol in \cite{MR20}.} Dolev-Strong designed an algorithm for the Byzantine broadcast (BB) problem, which can be converted into a Byzantine agreement algorithm with an initial round to broadcast the input Byzantine nodes under a public-key infrastructure and access to a global coin. The communication/bit complexity is $O(\kappa n^3)$ \cite{Fitzi02}. However, using multi-signature it can be improved to $O(\kappa n^2 + n^3)$, where $\kappa$ is a security parameter which is essentially the maximum size of the messages \cite{MR20}. %But the actual bit complexity is $O(\kappa n^3)$ \cite{Fitzi02}. 
While the original Dolev-Strong BB protocol tolerates $f \le n-1$ faults, the converted BA protocol works for the honest majority nodes, i.e., tolerates $f< n/2$ Byzantine faults which is optimal for an authenticated BA \cite{ADDN019,FM97,KK09,LSP82,SV17,MR20}. 

Dolev-Strong algorithm is deterministic and has a latency of $f+1$ rounds. The general idea of the BB protocol is to form a signature chain consisting of signatures from distinct nodes. A signature chain of $f+1$ signatures must contain a non-faulty signature from a non-faulty node which can send the value to all the other nodes. The protocol is designed in such a way that in $f+1$ rounds it forms a signature chain of size $f+1$. We adapted the Dolev-Strong BB protocol for the Byzantine agreement problem and used it in our implicit BA algorithm.   %For detailed analysis please go through the work of Dolev-Strong \cite{DS83}.

Let us now describe the implicit BA algorithm. The public key ($p_k$) of the nodes is known to all the nodes (as distributed by the trusted third party), but a node does not know which port or edge is connecting to which node (having a particular public key). To minimize the message complexity, a generic idea is to select a set of small-size candidate nodes, which will be responsible for solving the (implicit) agreement among themselves. Thus, it is important to have honest majority in the set of candidate nodes. For this, a random set of nodes, called as {\em candidate nodes} or {\em committee nodes}, of size $O(\log n)$ is selected. Let us denote the candidate nodes set by $\mathcal{C}$. A Byzantine node may try to claim that it is in $\mathcal{C}$, which needs to be stopped to guarantee the honest majority in $\mathcal{C}$. To overcome this problem, we take the help of a global coin and a keyed hash function, which together determine the candidate nodes.

 %with respect to all the public keys which are known to all the nodes.}
%Random number works as the center around which candidate nodes are selected. 
A common random number, say $r$, is generated with the help of the global coin. The random number should be large enough to use as the key to the hash function. Note that the random number is generated after the selection of the Byzantine nodes by the adversary. Every node uses its respective public key $(p_k)$ as the message and $r$ as the key of the hash function, say, $H$. That is they compute $H_r(p_k)$. For each $p_{k_i}$, we represent its hash value $H_r(p_{k_i})$ as $H_i$. Since the hash values are random (with high probability), the smallest $c\log n$ values among the $n$ hash values are chosen to be the candidate nodes, where $c$ is a suitable constant (to be fixed later). More precisely, a node $i$ with the hash value $H_i$ is in $\mathcal{C}$ if it is in the smallest $c\log n$ values of the set $\{H_i \, : \, i = 1, 2, \dots, n \}$.  
%from the set $\{|r - p_k(i)| \, : \, i = 1, 2, \dots, n \}$. 
Therefore, a node can easily figure out the candidate nodes since it knows the random number $r$, hash function and the public keys of all the nodes. However, the node does not know the ports connecting to them (as $KT_0$ model).   

Although a node knows all the candidate nodes (in fact, their public keys), it does not know the edges connecting to the candidate nodes, since the network is anonymous, i.e., $KT_0$ model. It can be known by the candidate nodes by sending a message to all the nodes,  but that will cost $n\log n$ messages. Since knowing each other is message expensive in this model, the candidate nodes communicate among themselves via some other nodes. For this, each candidate node randomly samples $\Theta(\sqrt{n\log n})$ nodes among all the $n$ nodes; call them as {\em referee nodes}. The reason behind sampling so many referee nodes is to make sure there is at least one common “non-faulty'' referee node between any pair of candidate nodes. The candidate nodes communicate with each other via the referee nodes. Notice that a node may be sampled as a referee node by multiple candidate nodes. It might happen that a Byzantine referee node may change the value of an honest candidate node before forwarding it to the candidate nodes. To avoid this, we take advantage of digital signature. Each referee node signs the input value before transmitting it to its referee nodes. As a digital signature helps to detect forging messages, each candidate node considers only the genuine messages received from the referee nodes. 

Thus, we have a small committee of nodes (i.e., $\mathcal{C}$) with the highly probable honest majority and the committee nodes can communicate via the referee nodes. Then we apply the Dolev-Strong BA protocol \cite{DS83} in the committee to achieve agreement. Let us now present the adapted Dolev-Strong algorithm to work for the Byzantine agreement.  
%the candidate nodes prioritize the value for the agreement which is in majority. If there is no such value then the value $0$ will be prioritized by default. This value will be send to all the nodes along with the sender's signature. If an input value is received in the $i^{th}$ round then it should possess a set of at least $i$ signatures (i.e., a chain of at least $i$ different signatures). This prevents the holding of higher priority value for indefinite rounds for the Byzantine nodes. There might be the case that a faulty node hold the value and send to only faulty nodes. Then also, after $c\log n$ rounds there exist at least one non-faulty node that will send the value to all the nodes. If highest priority value is initiated by non-faulty nodes then that value will reach in the very first iteration to all the nodes. These nodes will broadcast that value (along with their signature). After that there will be no communication till $c \log n$ iteration and all the nodes will agree on that value. We define the iteration as the number of rounds taken by the network to send the message from one candidate node to another. Due to CONGEST model, our algorithm possess message size $O(\log n)$. Therefore, an iteration may take more more than 1 round to send the message from one candidate node to other (via referee nodes). 
%The communication between any two candidate nodes take place via referee nodes in an iteration. Since every round is of size $O(\log n)$, therefore, an iteration may take more than 1 round.

\noindent \textbf{Step 0:} Each candidate node $u$ does the following in parallel. $u$ signs and sends its input value to its corresponding referee nodes, say, $\mathcal{R}_u$. Each referee node $w$ sends all the received values to its respective candidate nodes, say, $\mathcal{C}_w$. Therefore, the candidate nodes have the input values of all the candidate nodes. Now each candidate node proposes a value for the agreement based on the priority, along with all the signatures received corresponding to that input value. If there is a value that is sent by the majority of the nodes (i.e., more than $(c \log n)/2$ nodes), then that value gets the {\em highest priority}. In case of more than one majority, the value proposed by the maximum number of nodes gets the highest priority. There might be the case, two values are proposed by the same number of nodes; in that case, the larger input value gets the highest priority. If a candidate node has the highest priority input value, then it sends the value (after signing) to all the candidate nodes along with the received signatures for the highest priority value. Otherwise, the candidate node does not send anything. The reason of getting more than one majority value is that a Byzantine node may propose different values to different nodes. If no such majority value is received, then the candidate nodes decide on a default value, say, the minimum in the input value set. By sending the input value, we mean sending the input value along with all the signatures corresponding to that input value. %Basically, default value (i.e., $0$ in our protocol which can be anything from the input set) has the least priority and the value proposed by maximum number of node possess the highest priority. 

Then, the following two steps are performed iteratively\footnote{An iteration is the number of rounds required to send messages from one candidate node to all the candidate nodes via the referee nodes. Since we consider the CONGEST model, an iteration may take up to $O(\log n)$ rounds in our algorithm.}.
 
\noindent \textbf{Step 1:} If a candidate node $u$ receives a set of at least $i$ legitimate signatures in $i^{th}$ iteration with a higher priority value than its earlier sent value, then $u$ proposes this new highest priority value along with all the signatures to its referees nodes $\mathcal{R}_u$. 

\noindent \textbf{Step 2:} If a referee node $w$ receives a set of at least $i$ legitimate signatures in $i^{th}$ iteration (with the highest priority value) then $w$ forwards this highest priority value to its corresponding candidate nodes $\mathcal{C}_w$. 

The above two steps (i.e., Step~1 and~2) are performed for $O(c\log n)$ iterations and then the algorithm terminates. In the end, all the (non-faulty) candidate nodes have the same value, either the default value or the value possessed by the majority of nodes (the highest priority value). Intuitively, there are $O(c \log n )$ candidate nodes and a single node may propose the highest priority value in each iteration (say, the single node with the highest priority value is faulty and send to only faulty nodes) and thus the highest priority value may propagate slowly. But eventually, the highest priority value is received by all the candidate nodes as the set of candidate nodes contains the majority of the non-faulty nodes (see Lemma~\ref{lemma:atmostfaulty}) and there is a common non-faulty referee node between any pair of candidate nodes (see Lemma~\ref{lemma:common-referee}). On the other hand, if there is no highest priority value, then they decide on the default value. 
Thus, the (honest) candidate nodes agree on a unique value. 
%after $O(\log^2 n )$ rounds (as each iteration takes at most $O(\log n)$ rounds).  
A pseudocode is given in Algorithm~\ref{alg:agreement}.

\begin{algorithm}[ht]
\caption{\sc Authenticated-Implicit-BA}\label{alg:agreement}
\begin{algorithmic}[1]
\Require{A complete $n$ node anonymous network with $f \le (1/2-\epsilon)n$ Byzantine nodes. Each node receives an input value provided by an (static) adversary, a pair of public-private keys $(p_k, s_k)$, keyed hash function and a global coin. $\eps>0$ is a fixed constant.} 
\Ensure{Implicit Agreement.}
\Statex
\State Select $c \log n$ nodes as the candidate nodes set (say, $\mathcal{C})$, which have the smallest $c \log n$ hash values generated with the help of public key and random number. The value of the constant $c$ is $3 \alpha/ \eps^2$, follows from Lemma~\ref{lemma:atmostfaulty}. %\anis{this constant seems incorrect. check it. I think it would be$3\alpha/\eps^2$}

%\State Select the $c \log n$ nodes as candidate nodes set (say, $\mathcal{C})$, which are $c \log n$ public keys nearest to a random number generated via the global coin. The value of the constant $c$ is $(3-6\eps)/2\eps^2$, follows from Lemma~\ref{lemma:atmostfaulty}. 

\State Each candidate node $u$ randomly samples $2\sqrt{n \log n}$ nodes as referee nodes (say, $\mathcal{R}_u$).  

\State Each candidate node $u$ signs and sends its input value (with signature) to $\mathcal{R}_u$. %Sending input value conveys both input value and all the corresponding signature.

\State Each referee node $w$ sends all the received values to their respective candidate nodes $\mathcal{C}_w$ along  with the legitimate signatures one by one. It takes $O(\log n)$ rounds.

\State Each candidate node $u$ sends the input value along with all the received signatures to  $\mathcal{R}_u$ based on the (highest) priority. \Comment{Highest priority is defined in the description, Step~0.}
%Priority is given to the input value which is received from at least $c \log n/2$ nodes. In case of tie the most received value will get the higher priority. If there is no majority value then each node proposes the value $0$ (by default).

\For{the next $(c \log n)$ iterations, the candidate and referee nodes in parallel} 
    \State Each referee node $w$ checks:  %\Comment{Takes $\Theta(\log n)$ rounds.}
    \If{$w$ receives a set of at least $i$ legitimate signatures in the $i^{th}$ iteration}
        \State {$w$ sends highest priority value to $\mathcal{C}_w$.}
    \Else
        \State $w$ does not send any messages.
    \EndIf
    
    \State{Each candidate node $u$ checks:} %\Comment{Takes $\Theta(\log n)$ rounds.}
    \If{$u$ receives a set of at least $i$ legitimate signatures in the $i^{th}$ iteration with a highest priority value than it sent earlier}
        \State{$u$ sends highest priority value to $\mathcal{R}_u$.}
    \Else
        \State{$u$ does not send any messages.}
    \EndIf
\EndFor
\State{All the (non-faulty) candidate nodes have the same highest priority value on which they agree. Otherwise, if they do not receive any highest priority value, they agree on a default value, say, the minimum in the input value set.}
\end{algorithmic}
\end{algorithm}
%\end{figure}
%\subsection{Supporting Lemma}

Let us now show the above claims formally. We first show that the majority of the nodes in the candidate set are honest. 
%\vspace{-0.05cm}
\begin{lemma}\label{lemma:atmostfaulty}
The number of Byzantine nodes in the candidate set $\mathcal{C}$ is strictly less than $\frac{1}{2} |\mathcal{C}|$ with high probability. 
%The candidate set $\mathcal{C}$  contains at most $\frac{1}{2} |\mathcal{C}|$ Byzantine nodes with high probability.
\end{lemma}
%\vspace{-0.4cm}
\begin{proof}
Let $\alpha$ fraction of the nodes in the network are Byzantine where $\alpha = 1/2 - \epsilon$ is a fixed constant. So $\alpha + \epsilon = 1/2$. The nodes in $\mathcal{C}$ are chosen uniformly at random (i.e., with probability $1/n$) and the number of Byzantine nodes is at most $\alpha n$, the probability that a particular node in $\mathcal{C}$ is Byzantine is at most $\alpha$ . Let $\mathcal{C}$ contain $k$ nodes, $\{u_i \,|\, i= 1, 2, \dots, k\}$. Let us define random variables $X_i$s such that $X_i = 1$ if $u_i$ is Byzantine, and $0$ otherwise. Further, $X = \sum_{i=1}^{k} X_i$ is the total number of Byzantine nodes in $\mathcal{C}$. Then, by linearity of expectation, $E[X] \le \alpha  k$. Then, by Chernoff bounds \cite{MU04},

 \begin{align} 
\Pr(X \ge (\alpha + \eps) k) & = \Pr(X \ge (1 + \eps/\alpha) E[X]) \le \exp{\left(- E[X] (\eps/\alpha)^2 /3\right)} \notag\\ 
& \le  \exp{\left(- (k \eps^2)/(3\alpha) \right)} %\notag\\ 
%& \le 1/n. 
\text{for $k = |\mathcal{C}| = (3\alpha/\eps^2)\log n$} \notag  \notag
\end{align}

Thus,  $\Pr(X < (\alpha + \eps) |\mathcal{C}|) = \Pr(X < \frac{1}{2}|\mathcal{C}|) > 1-1/n$. In other words, $\mathcal{C}$ contains at most $ (1/2 - \delta)|\mathcal{C}|$ Byzantine nodes with high probability, for any fixed $\delta> 0$.
\end{proof}

The candidate nodes communicate with each other via the referee nodes, which are sampled randomly by the candidate nodes. The number of referee nodes is sampled in such a way that there must be a common referee node between every pair of candidate nodes (so that the candidate nodes can communicate) and at the same time keep the message complexity lower. In fact, we need to guarantee a stronger result. Namely, there must be a {\em non-faulty} common referee node for reliable communication. %The proof of the following lemma is deferred to the Appendix due to space limitation. 
%\vspace{-0.1cm}
\begin{lemma}\label{lemma:common-referee}
Any pair of candidate nodes have at least one common non-faulty referee node with high probability.
\end{lemma}
%\vspace{-0.1cm}
\begin{proof}
Let us consider two candidate nodes $v$ and $w$, and let $x_i$ be the $i^{th}$ node selected by $v$. Let the random variable $X_i$ be $1$ if $x_i$ is also chosen by $w$ and $0$ otherwise. Since the $x_i$s are chosen independently at random, the $X_i$s are independent.  So, $Pr[X_i=1] = \frac{ 2\sqrt{n\log n}}{ n}$. Hence, the expected number of (choice of) referee nodes that $v$ and $w$ haven in common are:
\begin{align}
    E[X]&= E[X_1]+E[X_2]+ \dots + E[X_{(2\sqrt{n\log n})}] (\textit{by linearity}) \notag\\
    & = 2\sqrt{n\log n}\cdot \frac{ 2\sqrt{n\log n}}{n} = 4 \log n
\end{align}
Thus, by using the Chernoff bound ~\cite{MU04},  $\Pr[X < (1-\delta)E[X]] < e^{-\delta^2E[x]/2}$ for $\delta = 0.5$, we get
\begin{align}
\Pr[X < (1-0.5) 4\log n] < e^{- (0.5)^2 (4\log n)/2} < \frac{1}{\sqrt{n}}
\end{align}
For that reason at least $2 \log n$ choice of nodes by $v$ are jointly chosen by $v$ and $w$. As a consequence, the probability that none of these choices is non-faulty:
\begin{equation}
    {\left(\frac{1}{2}- \eps\right)}^{2 \log n} < \frac{1}{n}
\end{equation}
 Therefore, the probability of selecting at least one non-faulty referee node from the sampled nodes of $u$ is at least $1-1/n$.
\end{proof}
%\end{lemma}

% \begin{proof}
%  Let us consider two candidate nodes, namely $u$ and $v$. We show that there is at least one common non-faulty referee node for $u$ and $v$, i.e., $\R_u \cap \R_v$ has at least one non-faulty node with high probability. Recall that each candidate node samples $2{\sqrt{n\log n}}$ referee nodes independently and uniformly among $n$ nodes. As there are at least $(1/2 + \epsilon)n$  non-faulty nodes, the probability of sampling any non-faulty node as referee node is: $(1/2+\epsilon)\cdot\frac{ 2\sqrt{n\log n}}{ n}$. Thus, the probability that the candidate node $u$ is not selecting a non-faulty node as referee node is: $(1 - (1/2+\epsilon)\cdot\frac{ 2\sqrt{n\log n}}{ n})$. It also holds for $v$. %\\
 
% Now, the candidate node $v$ samples $2{\sqrt{n\log n}}$ referee nodes. So, the probability of not selecting a 
% non-faulty referee node by $u$ from the sampled nodes of $v$ is (or vice-versa):
% \begin{align}
%   \left(1 - \cfrac{2 \sqrt{n\log n}}{n/(1/2+\epsilon)}\right)^{2\sqrt{n\log n}} 
%   &< \left(1 - \cfrac{\sqrt{n\log n}}{n}\right)^{2\sqrt{n\log n}} \notag\\
%   & \leq e^{-2\log n} = \cfrac{1}{n^2}
% \end{align}

% Therefore, the probability of selecting at least one non-faulty referee node from the sampled nodes of $u$ is at least $1-1/n^2$. Using the standard union bound over all the pairs, the claim holds for any pair of candidate nodes. 
% \end{proof}
%\vspace{-0.1cm}
Lemma~\ref{lemma:atmostfaulty} ensures that a committee (i.e., the candidate set) of size $O(\log n)$ with honest majority can be selected. Lemma~\ref{lemma:common-referee} ensures that the committee nodes can communicate with each other reliably via the referee nodes. Thus, the BA problem on $n$ nodes reduces to a $O(\log n)$-size committee nodes. Then the Dolev-Strong BA protocol ensures that the committee nodes achieve Byzantine agreement among themselves deterministically. Therefore, the algorithm (Algorithm~\ref{alg:agreement}) correctly solves the implicit Byzantine agreement with high probability (i.e., among the committee nodes only). The non-faulty nodes which are not selected in the committee may set their state as “undecided'' immediately after the selection of the candidate nodes.   

Below, we analyze the message and time complexity of the algorithm. 
%\vspace{-0.1cm}
\begin{lemma}\label{lem:message}
 The message complexity of the authenticated implicit BA algorithm is $O(n^{0.5} \log^{3.5} n)$.
\end{lemma}
%\vspace{-0.6cm}
\begin{proof}
In Step~3 and~4 of the algorithm, $O(\log n)$ candidate nodes send their input value with signature to the other candidate nodes via the $2 \sqrt{n \log n}$ referee nodes. Here the size of each message is $O(\kappa + \log n)$ bits, where $\kappa$ is the security parameter, the size of the signature. So these two steps uses $O(\sqrt{n \log n})\cdot O(\log n)\cdot O(\kappa +\log n) = O((\kappa +\log n) \sqrt{n \log^3 n})$ bits.

Inside the $O(\log n)$ iteration (Step~6): $O(\log n)$ candidate nodes may have at most $O(\log n)$ messages to be sent to the referee nodes for $O(\log n)$ rounds. The size of each message is $O(\kappa+\log n)$ bits. So it uses $O((\kappa+\log n) \sqrt{n\log^7 n} )$ bits. Further, the same number of message bits are used when the referee nodes forward the messages to the candidate nodes. Thus, a total $O((\kappa+\log n) \sqrt{n\log^7 n} )$ bits are used inside the iteration. 

Hence, the total communication complexity of the algorithm is $O((\kappa+\log n) \sqrt{n\log^{7} n})$ bits. Since the length of $\kappa$ is typically to be the maximum size of the messages \cite{MR20}, it is safe to assume $\kappa$ is of order $O(\log n)$ for large $n$. Therefore, the  total number of bits used is: $O(\sqrt{n\log^{9} n})$. Thus, the message complexity of the algorithm is $O(\sqrt{n\log^{7} n})$, since the size of each message is $O(\log n)$\footnote{Alternatively, the communication complexity is $O\left(\sqrt{n\log^{9} n}\right)$ bits.}.  
\end{proof}
%\vspace{-0.3cm}
\begin{lemma}\label{lem:time}
The time complexity of the algorithm is $O(\log^2 n)$ rounds.   
\end{lemma}
%\vspace{-0.6cm}
\begin{proof}
There are $O(\log n)$ candidate nodes, and each may have $O(\log n)$ messages to be sent to its referee nodes in parallel. Thus, it takes $O(\log^2 n)$ rounds, since it takes one round to send a constant number of messages of size $O(\log n)$ bits. The same time bound holds when the referee nodes forward the messages to the candidate nodes. Therefore, the overall round complexity is $O(\log^2 n)$.
\end{proof}

%Further, the algorithm correctly solves Byzantine agreement with high probability. 
Thus, we get the following result of the  implicit Byzantine agreement with  authentication. 
%\vspace{-0.1cm}
\begin{theorem}[Implicit Agreement] \label{thm:main-implicit}
Consider a synchronous, fully-connected network of $n$ nodes and CONGEST communication model. Assuming a public-key infrastructure with the keyed hash function, there exists a randomized algorithm which, with the help of a global coin, solves implicit Byzantine agreement with high probability in $O(\log^2 n)$ rounds and uses $O(n^{0.5} \log^{3.5} n)$ messages while tolerating  $f \le (1/2-\epsilon)n$ Byzantine nodes under non-adaptive adversary, where $\eps$ is any fixed positive constant.   
%For a non-adaptive adversary, there exists a randomized algorithm which solves implicit Byzantine agreement in $O(\log^2 n)$ rounds and uses $O(n^{0.5} \log^{4.5} n)$ messages while tolerating  $f \le (1/2-\epsilon)n$ Byzantine nodes. 
\end{theorem}

%\vspace{-0.1cm}
Our implicit agreement algorithm can be easily extended to solve explicit agreement (where all the honest nodes must decide on the same value satisfying the validity condition) in one more round. After the implicit agreement, the committee nodes send the agreed value (along with the signature) to all the nodes in the network in the next round. This incurs $O(n \log n)$ messages. Then all the nodes decide on the majority value since the majority of the nodes in the committee are honest. Thus, the following result of explicit agreement follows immediately.     
%\vspace{-0.05cm}
\begin{theorem}[Explicit Agreement]\label{thm:explicit}
Consider a synchronous, fully-connected network of $n$ nodes and CONGEST communication model. Assuming a public-key infrastructure with the keyed hash function, there exists a randomized algorithm which, with the help of global coin, solves Byzantine agreement with high probability in $O(\log^2 n)$ rounds and uses $O(n \log n)$ messages while tolerating  $f \le (1/2-\epsilon)n$ Byzantine nodes under non-adaptive adversary, where $\eps$ is any fixed positive constant.   
\end{theorem}

\onlyLong{
Let us now discuss some relevant follow-up results.% from the above construction.

% \subsection{Byzantine Agreement in Simple Graph \manish{new subsection}}
% Our candidate nodes from Algorithm \ref{alg:agreement}, can be helpful to reach at agreement in simple graph which does not possess any Byzantine node as articulation point. In case of Byzantine node as articulation point, the graph might behave as a disconnected graph. In this setting, the candidate nodes broadcast the signature and input value to all the neighbors (unlike Algorithm \ref{alg:agreement} which only sent to the referee nodes). These neighbors sends the received information to all other neighbors and eventually all the candidate nodes gets each other's input value. Thus, all the candidate nodes agree on the highest priority (as in Algorithm \ref{alg:agreement}). Here, the diameter of the graph is $D$. Each node can send one message (of size $O(\log n)$) via an edge at an instance, and there are $O(\log n)$ such messages with the signature size $O(\log n)$. Therefore, round complexity is $O(D \log n)$ since the message can be sent via pipeline. While at max $O(\log^2 n)$ message are sent from a given edge, therefore, message complexity is $O(m \log^2 n)$. Notice that a node does not broadcast any non-candidate node's message. Also, in case of candidate nodes' message, node broadcast the message only once.

\subsection{Byzantine Leader Election}\label{sec:le}
A leader can be elected without any communication in this model. Since the public keys are known to all the nodes, the node nearest to the random number generated through the global coin will be the leader. In case of a tie, the node with the largest public key will be the leader. The elected leader is non-faulty with probability at most $(1-f/n)$, where $f$ is the number of faulty nodes. The reason is that the Byzantine nodes can act as an honest node and the probability of electing a Byzantine node as leader is the same as for an honest node. Since there are at most $f \le (1/2-\epsilon)n$ Byzantine nodes, the elected leader is non-faulty with at most constant probability. 

We remark that a leader solves the implicit agreement, as the leader can agree on its own input value. It also solves explicit simply by sending its value to all the nodes. However, the success probability of the Byzantine agreement is at most constant. A Byzantine agreement protocol with a high success probability is important for practical solutions (which is presented in the above section).% But our goal is to solve BA with high success probability.   

\subsection{In the $KT_1$ Model}\label{sec: KT1}
In the $KT_1$ (\textbf{K}nown \textbf{T}ill \textbf{1}) communication model, implicit Byzantine agreement can be solved in $O(\log^2 n)$ rounds using $O(\log^3 n)$ messages in the same settings. First, select the candidate nodes of size $\Theta(\log n)$ as we selected in the $KT_0$ model (see Section~\ref{sec:agreement}). In the $KT_1$ model, each node is aware about the IDs' of its neighbor. Therefore, the candidate nodes know each other and the port connecting to them. In fact, they form a complete graph of size $\Theta(\log n)$ where one node knows the other. Now, we can run the Algorithm~\ref{alg:agreement} on this complete graph of candidate nodes, which essentially solves implicit Byzantine agreement in $O(\log^2 n)$ rounds and uses $O(\log^3 n)$ messages.  %Similarly, in case of without cryptographic assumptions by using the procedure used in $KT_0$ (without cryptographic assumption) one can select the candidate nodes and reach at agreement in $O(\log^3 n)$ message and $O(\log^2 n)$ round. Although message and round complexity are same in both the model, yet we gain the resilience with cryptographic assumptions i.e., $(1/2-\epsilon)n$ as compared to without cryptographic assumptions i.e., $(1/3-\epsilon)n$.

The leader election can be solved in the same way as in Section~\ref{sec:le}; there is no need of any communication. 
%In case of leader election, there is no need of communication/round. The node nearest to the random number generated through the global coin will be the leader. The probability with which a non-faulty node is leader is directly proportional to the fraction of non-faulty nodes in the network irrespective of the model's assumption, i.e, with or without cryptographic assumptions. 
%Although, one can elect the node with the highest ID among the candidate nodes as leader. But in that case, adversary might choose the highest ID node as Byzantine node before the execution start.

%\begin{remark}
    %In our model, we can replace the Global coin with minimum public key of the network. In that case, adversary selects the byzantine nodes before the trusted third party provides the $(p_k, s_k)$ pairs.
%\end{remark}

\subsection{Removing the Global Coin and Hash Function Assumption}\label{sec:wo-pki}
%It is possible to remove 
The assumption on accessing a global coin and hash function can be replaced by some other assumption (which might be stronger in some context). Previously, with the help of global coin a random set of candidate nodes is selected. Recall that the access of the global coin is given to the nodes after the adversary selects the Byzantine nodes (and it is a static adversary). Otherwise, the adversary can know which nodes would be in the candidate set and based on that selects the Byzantine nodes to be in the candidate set. Suppose there is no access of global coin (or shared random bits) and hash function. Then we need to make the adversary first selects the Byzantine nodes, and then the PKI setup is imposed in the network. Since the trusted third party generates a random pair of public-secret keys  $(p_k, s_k)$, the $O(\log n)$ nodes with the smallest public keys can be taken as the candidate nodes. Thus, the randomness in the candidate nodes set is preserved as required.  %Notice that the trusted third party is generating random pair of the keys with the help of appropriate protocol.
}

%\input{subset agreement}

%\vspace{-0.35cm}
\section{Lower Bound on Message Complexity}\label{sec: lower_bound}
%\vspace{-0.2cm}
We argue a lower bound of $\Omega(\sqrt{n})$ on the number of messages required by any algorithm that solves the authenticated Byzantine agreement under honest majority with high probability. Recall that all the nodes have access to an unbiased global coin. The nodes know the IDs of the other nodes, but are unaware of the port connecting to the IDs. Also, the communication is authenticated. %Our lower bound applies to all algorithms that send only $o(\sqrt{n})$ messages with probability at least $1 - 1/n$. 
%In other words, the result still holds for algorithms that have small but nonzero probability for producing runs where the number of messages sent is much larger (e.g., $\Omega(n)$). 

Our {\sc Authenticated-Implicit-BA}  algorithm solves multi-valued agreement in polylogarithmic rounds and uses $\Tilde{O}(\sqrt{n})$ messages (see, Theorem~\ref{thm:main-implicit}). In a non-Byzantine setting, the multi-valued agreement can be used to elect a leader by using the IDs of the nodes as the input values. Thus, any lower bound on the message complexity (and also on the time complexity) of the leader election problem also applies to the multi-valued agreement. Therefore, the $\Omega(\sqrt{n})$ lower bound shown by \cite{AMP18} for the leader election problem using global coin (in the non-Byzantine setting) also applies to our multi-valued Byzantine agreement with global coin. Note that the lower bound holds because of the high success probability requirement of the agreement; otherwise, the leader election \onlyLong{algorithm discussed in Section~\ref{sec:le}} solves the agreement with zero message cost, but with only constant success probability. However, the above argument does not hold for the binary agreement, where the input values are either $0$ or $1$. Below, we argue for the binary case.   

Let $A$ be an algorithm that solves the authenticated Byzantine (binary) agreement with constant probability (say, more than $1/2$) under an honest majority with the access of an unbiased global coin and uses only $o(\sqrt{n})$ messages. We show a contradiction. Recall that it is a $KT_0$ model; so nodes do not know which edge connects to which node-ID. To achieve agreement with constant probability, nodes must communicate with the other nodes; otherwise, if the nodes try to agree locally without any communication, it is likely that there exist two nodes that agree on two different values (this can be easily shown probabilistically). On the other hand, if all the nodes try to communicate, then the message complexity of $A$ would be $\Omega(n)$. So, only a few nodes need to initiate the process. Thus, $A$ must pick these few initiator nodes randomly; otherwise, if picked deterministically, the Byzantine nodes can take over the initiator nodes. The initiator nodes must communicate with the nodes in the network to achieve agreement. 

The initiator nodes do not know each other. In the $KT_0$ setting, for any two nodes to find each other with more than $1/2$ probability requires $\Omega(\sqrt{n})$ messages-- inferred from the following lemmas. %The proofs are deferred to the Appendix.% due to space limitation. %Thus, we get the following result on the lower bound of the message complexity.
%\vspace{-0.1cm}
\begin{lemma}\label{lem:know-each-other}
The $KT_0$ model takes $\Omega(\sqrt{n})$ messages for any two nodes to find each other with more than constant probability.
\end{lemma}
%\vspace{-0.2cm}
\begin{proof}
Suppose $x$ and $y$ be the two nodes. Both $x$ and $y$ samples $f(n)$ nodes uniformly at random from all the nodes to find out each other. Let $E$  be the event of having a collision in the random samples. Then, for $f(n) < \sqrt{n}$,
\begin{equation} \label{initiator bound equation}
\begin{split}
\Pr[E] = 1- \biggl(1- \frac{f(n)}{n}\biggr)^{f(n)} & = 1- e^{\frac{-(f(n))^2}{n}}  < 1- 1/e = 0.63
\end{split}
\end{equation}
%\manish{$$1-1/e = .63 \nleq 1/e (0.37)$$}
Therefore, $x$ and $y$ cannot find out each other with more than constant probability with $o(\sqrt{n})$ messages. 
\end{proof} 
% \vspace{-0.1cm}
Each initiator node samples some nodes. The initiator and the sampled nodes may exchange messages with the other nodes in the network throughout the execution of $A$-- all these nodes form a connected sub-graph. Let us call this sub-graph as {\em communication graph} of the initiator node. Thus, for every initiator node, there is a communication graph. Some of them may merge and form a single communication graph. However, it is shown in \cite{AMP18} that if the algorithm $A$ sends only $o(\sqrt{n})$ messages, then there exist two disjoint communication graphs w.h.p. (see, Section~2 of \cite{AMP18}). Since the nodes do not know each other (in $KT_0$), the communication graphs have similar information. The global coin also gives the same information to all the nodes. Since the communication graphs are disjoint, no information is exchanged between them. A formal proof of this argument can be found in \cite{KPPRT15, AMP18}.

Let $H_u$ and $H_v$ be the two disjoint communication graphs corresponding to the two initiators $u$ and $v$ respectively.  We show that $H_u$ and $H_v$ agree with opposite decisions, i.e., if $H_u$ decides on $0$ then $H_v$ decides on $1$ and vice versa.   
%\vspace{-0.01cm}
\begin{lemma}\label{lem:opposite-decision}
The nodes of $H_u$ and $H_v$ agree with opposing decisions. 
\end{lemma}
%\vspace{-0.45cm}
\begin{proof}
 The nodes in $H_u$ and $H_v$ are random since the network is anonymous and no information is known before contacting a node. Thus, the same $(1/2 -\eps)$ fraction of Byzantine nodes (as in the original graph $G$) are present in both $H_u$ and $H_v$ in expectation. Now consider an input distribution $\I$, in which each node in the graph $G$ is given an input value $0$ and $1$ with probability $1/2$. Then in expectation, half of the honest nodes in $H_u$ (and also in $H_v$) gets $0$ %(\manish{and half of them are Byzantine}) 
 and the other half gets $1$. %\manish{(also, half of them are Byzantine)}. 
 We argue that the two disjoint sets of nodes in $H_u$ and $H_v$ decide on two different input values under the input distribution $\I$. Recall that the nodes in $H_u$  have the same information as the nodes in $H_v$. Further, they run the same algorithm $A$. The Byzantine adversary (which controls the Byzantine nodes) knows the algorithm and also the input values of the honest nodes in $H_u$ and $H_v$. Since the number of Byzantine nodes in each of the $H_u$ and $H_v$ is almost half of their size, and $0-1$ distribution among the honest nodes is almost 50-50, the Byzantine nodes can control the output of the agreement in the two groups--- $H_u$ and $H_v$. 
 Suppose the Byzantine nodes in $H_u$ exchange some input bits with honest nodes to agree on $0$ in $H_u$, then the Byzantine nodes in $H_v$ must exchange opposite bits to agree on $1$ in $H_v$. Thus, $H_u$ and $H_v$ agree with opposing decisions. 
\end{proof}
 %\vspace{-0.01cm}
The above lemma contradicts the assumption that $A$ solves the Byzantine (binary) agreement with only $o(\sqrt{n})$ messages. The global coin does not help, as it gives the same information to all the nodes. Thus, we get the following result on the lower bound of the message complexity.
%\vspace{-0.05cm}
\begin{theorem}\label{thm: lower_bound}
Consider any algorithm $A$ that has access to an unbiased global coin and sends at most $f(n)$ messages (of arbitrary size) with high probability on a complete network of $n$ nodes. If $A$ solves the authenticated Byzantine agreement under honest majority with more than $1/2$ probability, then $f(n) \in \Omega(\sqrt{n})$. %This holds even if nodes are equipped with unique identifiers (chosen by an adversary).
\end{theorem}

%\vspace{-0.05cm}
Note that the lower bound holds for the algorithms in the {\em LOCAL} model, where there is no restriction on the size of a message that can be sent through edges per round \cite{Pelege2000}.  %which implies the same for the {\em CONGEST} model.

\vspace{-0.3cm}
\section{Conclusion and Future Work}\label{sec:conclusion}
\vspace{-0.2cm}
We studied one of the fundamental problem in distributed networks, namely Byzantine agreement. We showed that implicit Byzantine agreement can be solved with sublinear message complexity in the honest majority setting with the help of cryptographic set up of PKI and hash function, and access to a global coin. The bound is also optimal up to a $\polylog n$ factor. To the best of our knowledge, this is the first sublinear message bound result on Byzantine agreement. \onlyLong{We also implemented our algorithm to show its efficiency w.r.t different sizes of the Byzantine nodes. We further analyzed some relevant results which immediately follow from our main result. We also studied subset agreement, a generalization of the implicit agreement.}  %These results are summarized in Table~\ref{tbl:result}.  %From Table \ref{tbl:result}, in $KT_0$ model with and without cryptographic assumption message complexity are $\Tilde{O}(\sqrt{n})$ and $\Tilde{O}(n)$ respectively. While in $KT_1$ model, message complexity is $\Tilde{O}(1)$. Although round complexity is similar in all the four model i.e., $\Tilde{O}(1)$. 

% \begin{table*}[th]
% \center
% \begin{tabular}{|P{2.0cm}|P{3.00cm}|P{3.00cm}|P{5.0cm}|P{3.0cm}|}
% %\begin{tabular}{|P{0.55cm}|P{1.80cm}|P{1.20cm}|P{1.45cm}|P{1.65cm}|}
% %\begin{tabular}{|P{1.5cm}|P{2.50cm}|P{1.8cm}|P{3.5cm}|P{2.5cm}|}
% \hline
% \multicolumn{5}{|c|} {Summary of the Results}\\
% \hline
%  Model & Message & Round  & Cryptographic Assumptions & Resilience  \\
% \hline
% &&&&\\
% $KT_0$      & $O(n^{0.5}\log^{3.5}n)$ & $O(\log^2 n)$  & Yes & $f\leq (\frac{1}{2}-\epsilon)n$  \\
% &&&&\\
% %KT0         & Zero & Zero         & Yes & $f\leq (n-1)$ & Leader Election \\
% %&&&&&\\
% $KT_0$      & $O(n \log n)$ & $O(\log^2 n)$  & No &  $f\leq (\frac{1}{3}-\epsilon)n$ \\
% &&&&\\
% $KT_0$      & $\Omega(n^{0.5})$ & ---  & Yes & $f\leq (\frac{1}{2}-\epsilon)n$\\
% &&&&\\
% \manish{$KT_1$}      & $O(\log^3 n)$ & $O(\log^2 n)$  & Yes & $f\leq (\frac{1}{2}-\epsilon)n$ \\
% &&&&\\
% %KT1      & Zero & Zero  & - & - & Leader Election \\
% %&&&&&\\
% \manish{$KT_1$}      & $O(\log^3 n)$ & $O(\log^2 n)$  & No & $f\leq (\frac{1}{3}-\epsilon)n$ \\
% &&&&\\
% \hline

% \end{tabular}
%     \caption{Message and round complexity of Byzantine agreement with different assumptions.}
% \label{tbl:result}
% \end{table*}

A couple of interesting open problems are: (i) is it possible to achieve a sublinear message complexity Byzantine agreement algorithm without the global coin or hash function in this setting? (ii) whether a sublinear message bound is possible under adaptive adversary, which can take over the Byzantine nodes at any time during the execution of the algorithm?  %(II) Show a non-trivial lower bound of the implicit agreement. We do not know whether the $\tilde{O}(\sqrt{n})$ message bound is optimal or not, but we believe it is. %There exists a $\tilde{O}(n^{0.4})$ message bound agreement algorithm without Byzantine nodes but with the access of a global coin \cite{AMP18}. %(2) Extend the study of message complexity (and also time complexity) of the problems in general graphs. (3) Finally, whether the similar bounds can be achieved without global coin.

\bibliographystyle{plainurl}
\bibliography{DCS,security}
\end{document}